\documentclass[11pt]{article}

\usepackage[a4paper, total={6in, 10in}]{geometry}

\usepackage{preamble/packages}
\usepackage{preamble/commands}
\usepackage{amsmath,amssymb,amsthm}
\usepackage{booktabs}
\usepackage{microtype}

\newtheorem{theorem}{Theorem}[section]
\newtheorem{proposition}[theorem]{Proposition}
\newtheorem{lemma}[theorem]{Lemma}

\theoremstyle{definition}

\theoremstyle{remark}
\newtheorem{remark}[theorem]{Remark}
\newtheorem{example}[theorem]{Example}

\newcommand{\GL}{\mathrm{GL}}
\newcommand{\SL}{\mathrm{SL}}
\newcommand{\SO}{\mathrm{SO}}
\newcommand{\Or}{\mathrm{O}}
\newcommand{\SU}{\mathrm{SU}}
\newcommand{\Sp}{\mathrm{Sp}}
\newcommand{\C}{\mathbb{C}}
\newcommand{\Z}{\mathbb{Z}}
\newcommand{\Q}{\mathbb{Q}}
\newcommand{\lam}{\lambda}

\hypersetup{
    colorlinks=false,       
    pdfborder={0 0 1},      
    linkbordercolor={1 0 0}, 
    citebordercolor={0 1 0}, 
    urlbordercolor={0 1 1}   
}
\begin{document}

\begin{titlepage}
   \begin{center}
       \vspace*{2cm}
        \Huge
       \textbf{Decomposing Grassmann Monomials for Superfield Expansions}
       
       \vspace{0.3cm}
         \LARGE
       \vspace{0.2cm}
     
\begin{center}

{\bf
Jesse Woods
} \\
\vspace{2mm}

{\small{\it 
	Albert Einstein Center for Fundamental Physics,\\
	Institute for Theoretical Physics, University of Bern,\\
	Sidlerstrasse 5, CH-3012 Bern, Switzerland}
}
\vspace{2mm}
~\\
\texttt{jesse.woods@unibe.ch}\\
\vspace{1mm}
\end{center}
   \end{center}

   \vspace{1cm}
   \begin{abstract}
\baselineskip=14pt
Superspace extends spacetime by anticommuting Grassmann coordinates, whose indices may transform under spin and flavour groups. Decomposing Grassmann monomials into independent invariant structures is a central step in constructing superfield expansions, but becomes increasingly difficult in extended superspace. We present a general representation-theoretic procedure for decomposing the space $\Lambda^n(\mathbb{C}^{d_S}\otimes \mathbb{C}^{d_F})$, where $d_S$ and $d_F$ are the dimensions of the spin and flavour representations. The branching multiplicities for
$\GL(m)\downarrow{\Sp(m)}$ and $\GL(m)\downarrow{\SO(m)}$ are computed by
exact Weyl-character comparison at deterministic rational sample points, avoiding Littlewood-Richardson modification rules. We
establish completeness of the character ansatz and analyse the computational complexity, showing that it is polynomial in the
number of partitions of $n$ and in the group rank, without dependence on
Littlewood-Richardson combinatorics. We validate the resulting decompositions via dimension checks. We provide an accompanying standalone Julia script which implements this algorithm for groups of arbitrary rank, and produces the explicit contractions with invariant tensors.
\end{abstract}
\newpage

\vspace{5mm}
\end{titlepage}

\tableofcontents

\section{Introduction}
Superspace has been a powerful tool for the study of supersymmetric field theory since its inception\cite{Salam:1976ib}. By extending ordinary spacetime $\mathcal{M}$ parametrised by $\{x_\mu\}_{\mu=0,...\dim{\mathcal{M}}-1}$ by a set of anticommuting coordinates $\{\theta\}$, one can encapsulate the field content of the theory in terms of superfields $\Phi(x,\theta)$ defined over this supermanifold. In extended supersymmetry, these Grassmann coordinates
$\theta^i_\alpha$ can carry both a spinor index $\alpha=1,\dots,d_S$ and an
internal (R-symmetry or flavour) index $i=1,\dots,d_F$.

Taylor expansions of superfields in powers of anticommuting coordinates terminate after a finite $N+1$ number of terms where $N=d_S\cdot d_F$ due to the nilpotency of the Grassmann numbers, 
\begin{equation}
    \Phi(x,\theta)=\sum_{n=0}^{N} C^{\alpha_1...\alpha_n}_{i_1...i_n}(x)\theta_{\alpha_1}^{i_1}...\theta_{\alpha_n}^{i_n},
\end{equation}
with each coefficient corresponding to an ordinary spacetime component field, thereby making the structure of supersymmetric multiplets explicit and facilitating practical calculations of supersymmetric actions, transformations, and interactions.

For a single Majorana or Weyl spinor $\theta_{\alpha}$, the expansion terminates quickly and is written by
hand, and the field content is transparent. For extended superspace, the expansion has $d_Sd_F+1$ terms, and the number of
independent components at order $n$ is $\binom{d_S d_F}{n}$, with the full expansion having $2^{d_Sd_F}$ components. These independent components must be identified and interpreted. The practical question is not merely to enumerate components but to organise them: if one can write each monomial as a sum of terms in which invariant tensors of the symmetry groups of the indices absorb as many indices as possible of the coefficient fields, then it is easier to identify the field content of the theory, even if the resulting superfield expansion is longer than $N+1$ terms. Such organised expansions are what one actually uses when constructing actions, imposing shortening conditions, or matching component multiplets. 

To illustrate this, consider the chiral superspace associated to $4D$ $\mathcal{N}=2$ supersymmetry, given by coordinates $(x_\mu,\theta_\alpha^i)$, where $\alpha=1,2$ is an $\SL(2,\C)$ spinor index and $i=1,2$ is an $\SU(2)_R$ flavour index. Each group has an associated invariant totally antisymmetric tensor $\epsilon_{\alpha\beta}$ and $\epsilon_{ij}$ respectively. In the superfield expansion, consider the $\mathcal{O}(\theta^2)$ piece:
\begin{equation}
    C^{\alpha\beta}_{ij}(x)\theta^i_\alpha\theta^j_\beta.
\end{equation}
Here, $C^{\alpha\beta}_{ij}$ has $6$ independent components. However, by defining the symmetric objects
\begin{equation}
    \theta_{\alpha\beta}=\epsilon_{ij}\theta_\alpha^i\theta_\beta^j, \quad \theta^{ij}=\epsilon^{\alpha\beta}\theta_\alpha^i\theta_\beta^j,
\end{equation}
one can decompose the term into
\begin{equation}
    C^{\alpha\beta}(x)\theta_{\alpha\beta}+ C_{ij}(x)\theta^{ij},
\end{equation}
where $C^{\alpha\beta}(x)$ and $C_{ij}(x)$ are symmetric fields, each with three independent components, corresponding to a self-dual 2-form singlet, and an $\SU(2)$ triplet scalar. Thus, by using this basis the field content is more transparent.
Although an irreducible representation decomposition provides a useful
classification of the component fields, it is not always the most useful information for explicit superspace calculations. Beyond these, what is more useful is an explicit
decomposition in terms of contractions with invariant tensors. This becomes particularly advantageous for higher extended supersymmetry, where the large R-symmetry groups generate increasingly complicated tensor products. Working directly with a list of irreducible representations can be hard to compute with, whereas an invariant-tensor basis retains manifest covariance while keeping the component fields in a form that is practical.

There is much literature on the decompositions of such monomials, both by hand and computationally. Early works established the problem and
solved it case by case, see \cite{Rittenberg:1981cp,Siegel:1981ec,Howe:1981xy,Koller:1982cs} for a by no means comprehensive list. There has also been recent interest with the
availability of computer algebra, directed additionally at the higher-dimensional cases \cite{JamesGates:2019bwu,JamesGates:2020hrl,Gates:2020qtw,Gates:2020gtu,Mansouri:2022qug}. We note also that similar representation theory algorithms to ours exist across a variety of packages, see for example \cite{van1992lie,Feger:2019tvk,Fonseca:2020vke}.

This paper describes an algorithm, and an accompanying native implementation in Julia, that produces this organised expansion automatically for $d_S$-dimensional spinor indices and dimension $d_F$ flavour indices in the fundamental representation, with
a choice of structure group on each index from
$\{\GL(m)$, $\mathrm{U}(m)$, $\SL(m)$, $\mathrm{SU}(m)$, $\Sp(m)$, $\Or(m)$, $\SO(m)\}$. Thanks to the exceptional isomorphisms, this is sufficient to describe spinors in six dimensions or lower. For $D>6$ spacetime dimensions, the product of spinors requires more than just a $\GL(m)$ decomposition. Recent works due to
Gates, Hu and Mak~\cite{JamesGates:2019bwu,JamesGates:2020hrl,Gates:2020qtw} and
Mansouri~\cite{Mansouri:2022qug}, treat from $4$- up to $11$-dimensional superspaces by methods involving branching from $\SU(d)$ and adinkras to tabulate
the Lorentz content of scalar superfields. These work with the spinor index
directly, branching to the Lorentz algebra through its spinor representation,
which lies outside the present scope of this work. A strand of that
work \cite{Gates:2020gtu} addresses the same fundamental-embedding branching
considered here, which we mention in Section~\ref{sec:witnesses}.

Beyond determining the multiplicities and irreducible representations, we implement an additional ``witness''-layer, which produces the required contractions with invariant tensors explicitly, which are  objects of interest to a physicist. This algorithm was designed for inclusion inside the larger Julia computer algebra package ``Alakazam'' which has other superspace functionality\cite{woods2026alakazam}.

The only representation-theoretic input the algorithm needs is the support of the branching, i.e. which labels can occur. Everything after that is exact linear algebra, which either returns the unique answer or fails visibly. There are three key properties of this design:
\begin{enumerate}
\item \emph{Exactness.}  All arithmetic is over $\Q$ with arbitrary-precision
integers; there are no floating point numbers in the core functions, and no randomness anywhere. Thus, numerical
error is excluded by construction and results are reproducible run to run.
\item \emph{Self-certification.}  Representation-theoretic bookkeeping is
notoriously easy to get subtly wrong (stability ranges, modification rules and folding,
doublings from self-dual and anti-self-dual representations).  Rather than trusting any single derivation, the
program re-verifies each branching it computes - exact reproduction of the
character at every sample point and at an independent one, integrality of all
multiplicities, and an overall dimension sum rule - and refuses to return a
result that fails any check.  This follows the philosophy of certifying
algorithms~\cite{McConnell2011}.
\item \emph{Scalability in the group size.}  The cost should depend on the rank, rather than the
dimension of the representation being decomposed; in particular no
Littlewood-Richardson coefficients are ever computed for the purpose of
determining a multiplicity in the decomposition.
\end{enumerate}

Section~\ref{sec:setup} fixes notation and states the problem formally.
Section~\ref{sec:cauchy} recalls the Cauchy decomposition underlying the
first layer of the algorithm.  Section~\ref{sec:branching} develops the
character-comparison branching, proves completeness of the candidate ansatz
(Proposition~\ref{prop:parity} isolates the odd-rank parity phenomenon), and
proves that a successful run returns the true branching
(Proposition~\ref{prop:certified}).  Section~\ref{sec:witnesses} constructs
explicit contraction witnesses, including the $\eps$ and $\omega^{\wedge k}$ strategies needed for
$\SO$ and $\Sp$.  Section~\ref{sec:complexity} analyses
complexity; Section~\ref{sec:validation} reports validation; and
Section~\ref{sec:limitations} discusses limitations.

\section{Statement of the Problem}
\label{sec:setup}

Let $V=\C^{d_S}$ and $W=\C^{d_F}$, and let $\theta^i_\alpha$
($\alpha=1,\dots,d_S$; $i=1,\dots,d_F$) be generators of the Grassmann algebra
$\Lambda(V\otimes W)$.  A superfield is an element
\begin{equation}
\Phi(x,\theta)\;=\;\sum_{n=0}^{d_S d_F}\;
C_{(n)}{}^{\alpha_1\cdots\alpha_n}_{\,i_1\cdots i_n}(x)\;
\theta^{i_1}_{\alpha_1}\cdots\theta^{i_n}_{\alpha_n},
\end{equation}
where Einstein summation notation is implied on repeated indices. Here, each coefficient may, without loss of generality, be taken totally
antisymmetric under simultaneous exchange of $(\alpha_k,i_k)$ pairs, so that
the order-$n$ coefficient space is $\Lambda^n(V\otimes W)^*$ of dimension
$\binom{d_S d_F}{n}$.

Fix structure groups $G_S\leq \GL(V)$ and $G_F\leq \GL(W)$ from the
list $\{\GL$, $\mathrm{U}$, $\SL$, $\mathrm{SU}$, $\Sp$, $\Or$, $\SO\}$.  By an \emph{independent structure} at order $n$ we mean a term
\begin{equation}
C\,\cdot\,T\,\cdot\,\theta^{i_1}_{\alpha_1}\cdots\theta^{i_n}_{\alpha_n},
\end{equation}
where $T$ is an explicit product of $G_S$- and $G_F$-invariant tensors
contracted against a subset of the $\theta$ indices (possibly sharing further
contracted indices with $C$), and $C$ is a tensor in a single irreducible
representation of $G_S\times G_F$ (with the exception of self-dual and anti-self-dual decomposition as we discuss in Lemma~\ref{lem:pair}), such that the components taken over all
structures span the order-$n$ coefficient space and their coefficient
dimensions add up to $\binom{d_S d_F}{n}$ exactly. The invariant tensors for the groups that we will consider can be seen in Table \ref{tab:invars}.
\begin{table}[H]
\centering
\begin{tabular}{c|c}
\hline
Group & Invariant tensors 
\\
\hline
$\mathrm{GL}(m)$
& $\delta^i{}_j$
\\

$\mathrm{U}(m)$
& $\delta^i{}_j$
\\

$\mathrm{SL}(m)$
& $\delta^i{}_j,\ \epsilon_{i_1\cdots i_m}$
\\

$\mathrm{SU}(m)$
& $\delta^i{}_j,\ \epsilon_{i_1\cdots i_m}$
\\

$\mathrm{Sp}(m)$
& $\omega_{ij}=-\omega_{ji}$, $\omega^{\wedge k}$, $k=1,...,m/2$
\\

$\mathrm{O}(m)$
& $\eta_{ij}$
\\

$\mathrm{SO}(m)$
& $\eta_{ij},\ \epsilon_{i_1\cdots i_m}$
\\
\hline
\end{tabular}
    \caption{Invariant tensors for common groups}
    \label{tab:invars}
\end{table}

Our focus is on the independent structures in terms of $\eps,$ $\eta,$ and $\omega$, depending on the group. $\delta$ will play no role here.
The algorithmic goal is to produce the list of independent structures, with
each $C$ of the smallest possible dimension, together with a machine-checked
proof of the dimension sum rule at every order.

\section{The Cauchy Decomposition}
\label{sec:cauchy}

The first layer is classical~\cite{FultonHarris,Macdonald,Weyl}: as a
$\GL(V)\times\GL(W)$ module,
\begin{equation}
\label{eq:cauchy}
\Lambda^n(V\otimes W)\;\cong\;\bigoplus_{\lambda}
\;\mathbb{S}_\lam(V)\otimes\mathbb{S}_{\lam'}(W),
\end{equation}
where $\lambda$ is an integer partition of $n$ with at most $d_S$ rows and at
most $d_F$ columns characterising a Young tableau, $\lam'$ denoting the conjugate partition corresponding to the transposed Young tableau, and
$\mathbb{S}_\lam$ the Schur functor which assigns the irreducible $\GL$ representation obtained by symmetrising the rows and antisymmetrising the columns of the associated tableau. 

For a simple example, consider $\Lambda^2(V\otimes W)$ with $V=\mathbb{C}^{2}$, $W=\mathbb{C}^{2}$ and the basis $\theta^i_\alpha\theta^j_\beta$. There are two possible $\lambda$ coming from the partitions of ``$2$'' given by $(2)$ and $(1,1)$. Then, the space decomposes as
\begin{equation}
    \Lambda^2(V\otimes W)
\simeq
\left(\underbrace{\mathrm{Sym}^2(V)}_{\lambda=(2)}\otimes\underbrace{\Lambda^2(W)}_{\lambda'=(1,1)}\right)
\oplus\left(
\underbrace{\Lambda^2(V)}_{\lambda=(1,1)}\otimes\underbrace{\mathrm{Sym}^2(W)}_{\lambda'=(2)}\right)
\end{equation}
which corresponds to the decomposition
\begin{equation}
    \theta^i_\alpha\theta^j_\beta=\theta^{[i}_{(\alpha}\theta^{j]}_{\beta)}+\theta^{(i}_{[\alpha}\theta^{j)}_{\beta]}.
\end{equation}

Generally, a convenient basis vector of each
summand is fixed by a pair of Young tableaux of shapes $\lam$ and $\lam'$
filled with the labels $1,\dots,n$ (one label per $\theta$ copy).

When $G_S=\mathrm{SL}(V)$ or $\mathrm{SU}(V)$, a column of full height $d_S$ in the spinor tableau corresponds to the invariant tensor
$\epsilon^{\alpha_1\cdots\alpha_{d_S}}$ and can therefore be contracted to a singlet. This allows full columns of height $d_S$ to be removed from the Young diagram. For $G_S=\mathrm{GL}(V)$ or $\mathrm{U}(V)$, the corresponding $\epsilon$-tensor transforms with the determinant and is not strictly an invariant, so no such reduction is possible. For $G\in\{\GL,\mathrm{U},\SL,\mathrm{SU}\}$, this is the only reduction available. The decomposition of $\mathrm{U}$ and $\mathrm{SU}$ reduces to the cases of $\GL$ and $\SL$ respectively.

For $G\in\{\Sp,\Or,\SO\}$ the summands of~\eqref{eq:cauchy} are reducible as \emph{$G$-representations} and
the second layer of the algorithm decomposes them.

\section{Branching by Exact Character Comparison}
\label{sec:branching}
Branching describes how a representation of a bigger group splits into representations of a smaller group. Given a representation $\rho$, the \emph{character} $\chi$ of the representation is the trace
\begin{equation}
    \chi_\rho(g)=\Tr \rho(g).
\end{equation}
Since the trace is a homomorphism, this satisfies
\begin{equation}
\chi_{\rho_1\oplus\rho_2}=\chi_{\rho_1}+\chi_{\rho_2}
\end{equation}
Thus, if we expect a branching, its character is the sum of the
characters of its constituents weighted by their multiplicities. This turns
the decomposition into a linear problem, where the multiplicities are the unknowns. Since the characters are given in closed form by the Weyl character formula, none of the representation spaces, their projectors, or any basis ever has to be constructed.

\subsection{The Ansatz and Completeness}
We now describe how the restriction of a $\GL(m)$ representation to
$G=\Sp(m)$ or $\SO(m)$ is determined in practice.
Let us take $m=d_S$ or $d_F$, and set
\begin{equation}
    r=\left\lfloor\frac{m}{2}\right\rfloor.
\end{equation}
Irreducible representations of $\Sp(m)$ and $\SO(m)$ are labelled by
partitions $\mu$ with $\ell(\mu)\le r$, where $\ell(\mu)$ is the length of the partition $\mu=(\mu_1,...\mu_{\ell(\mu)})$. We denote the corresponding
characters by $\chi^G_\mu$.

For $\SO(2r)$, a label with $\ell(\mu)=r$ and $\mu_r>0$ corresponds
to two irreducible representations, conventionally denoted
$\mu_+$ and $\mu_-$, which are exchanged by the outer automorphism of $\SO(2r)$\cite{FultonHarris}. We can write
\begin{equation}
\chi^G_\mu=\chi_{\mu_+}+\chi_{\mu_-}
\end{equation}
in this case, and similarly count the combined dimension of the two
representations. This convention is natural for the restrictions considered
here because both representations always occur with the same multiplicity:

\begin{lemma}
\label{lem:pair}
In the restriction of any finite-dimensional $\GL(2r)$ representation to $\SO(2r)$, the
irreducible representations $\mu_+$ and $\mu_-$ occur with equal multiplicity.
\end{lemma}
\begin{proof}
Suppose the restriction contains $a\mu_+\oplus b\mu_-$.
Conjugation by any improper orthogonal transformation $g_0\in\Or(2r)$ is an outer automorphism of
$\SO(2r)$
\[g\mapsto g_0gg_0^{-1}\]
interchanging $\mu_\pm$. Since the original representation is a $\GL(2r)$ representation, conjugation gives an equivalent representation. After restricting to $\SO(2r)$, the conjugated
representation therefore has the same decomposition as the original one,
but with $\mu_+$ and $\mu_-$ exchanged. Hence
\[
a\mu_+\oplus b\mu_-
\cong
b\mu_+\oplus a\mu_-,
\]
which implies
\[
a=b.
\]
\end{proof}
The lemma fixes the convention for the two associated representations of
$\SO(2r)$, corresponding to self-dual and anti-self-dual pairs. 

We now turn to the actual branching problem. Given a
$\GL(m)$ representation $\mathbb S_\lambda(\mathbb C^m)$, which
$\SO(m)$ or $\Sp(m)$ representations can occur when it is restricted to
the subgroup?

The standard answer is provided by the Littlewood restriction rules\cite{Littlewood,Koike1987}. These express the
multiplicities of $\mu$ in the restriction $\mathbb{S}_\lam(\C^m)|_G$ explicitly as a sum of
Littlewood-Richardson coefficients $c^{\lam}_{\mu\delta}$ (which are the multiplicities of the corresponding representation $\mathbb{S}_\lambda$ in the tensor product $\mathbb{S}_\mu\otimes \mathbb S_\delta$), over partitions
$\delta$ with all rows even (orthogonal case) or all columns even
(symplectic case), provided $\ell(\lam)\le r$, which is the ``stable range''. 

Outside of these cases, the naive sums must be corrected by King's modification
rules\cite{King1975}. These ``fold'' inadmissible labels back to admissible ones. For example, a candidate label $\nu$ whose length $\ell(\nu)$ exceeds the rank, has representations isomorphic to those of label $\Tilde{\nu}$, obtained by replacing the first column with one of height $m-\ell(\nu)$. These differ by a contraction with the $\eps$ tensor. Symplectic labels behave similarly, using wedge powers of $\omega$.

Our algorithmic implementation sidesteps these rules, and determines the
multiplicities as the unique solution of a linear system of character values. Folding enters below only as a proof device: Proposition~\ref{prop:parity} uses
it to bound which labels can occur, but the decomposition algorithm never performs a fold. For this approach, however, we still need to know in advance which irreducible representations can occur. The next
proposition gives precisely this candidate set.

\begin{proposition}[Completeness of the candidate set]
\label{prop:parity}
Let $\lam$ be a partition of $n$ with $\ell(\lam)\le m$, and let $c_\mu$ denote the
multiplicity of $\mu$ in the restriction of $\mathbb{S}_\lam(\C^m)$ to
$G$.  Then $c_\mu=0$ unless $\ell(\mu)\le r$ and
\begin{equation}
n-|\mu|\;\in\;
\begin{cases}
2\Z_{\ge0}, & G=\Sp(2r)\ \text{or}\ G=\SO(2r),\\[2pt]
\Z_{\ge0}, & G=\SO(2r+1).
\end{cases}
\end{equation}
For $G=\SO(2r+1)$, both even and odd values of $n-|\mu|$ can occur.
\end{proposition}
\begin{proof}
The key observation is that the invariant bilinear form $\eta_{ij}$ or $\omega_{ij}$ allows pairs of
indices to be contracted.
Contraction with the form removes boxes in pairs from the Young tableau, so before anything else we have
\[
|\mu|=n-2k
\]
for some $k\geq0$. For $\Sp(2r)$ and $\SO(2r)$, the subsequent folding of labels that are too
long changes the number of boxes only by an even amount. So, the parity is preserved and 
\[n-|\mu|\in 2\mathbb{Z}_{\geq0}.\]
For $\SO(2r+1)$, folding a label can change the number of boxes by an odd amount as the dimension $2r+1$ is odd, meaning either parity is possible. For example, when $r=1$, $m=3$, the naive label is $\nu=(1,1)$, but this has column height $\ell(\nu)=2$, which is bigger than $r=1$. Thus, we fold to a column of height $m-\ell(\nu)=1$, giving the resulting folded label $\mu=(1)$. This is the familiar isomorphism between cross products and vectors in $3D$,
\[
\Lambda^2\mathbb C^3\simeq\mathbb C^3,
\]
for which $n=2$ and $|\mu|=1$, so that
\[
n-|\mu|=1
\]
is odd. This explicitly demonstrates that odd values must be included in the
candidate set.

\end{proof}
This completeness statement is important computationally. If the candidate
set is incomplete, the exact character equations need not have a solution.
Thus the character system itself provides a useful consistency check: an
omitted representation manifests as an inconsistency rather than
producing an incorrect decomposition.

\subsection{Determination and Run-time Certification}

Let us fix the candidate list $M=\{\mu\}$ of Proposition~\ref{prop:parity}
(cardinality $|M|\le\sum_{j\le n}p(j)$, $p$ the partition function). We would like to determine how many times $c_\mu$ each candidate appears,
\begin{equation}\label{eq:problem}
\mathbb{S}_\lam(\C^m)|_G=\bigoplus_{\mu\in M}c_\mu \rho_\mu.
\end{equation}
This amounts to solving the equation
\begin{equation}
    \chi_\lambda(g)=\sum_{\mu\in M}c_\mu \chi_\mu(g).
\end{equation}
This is a linear equation for each group element, and to solve for $c_\mu$ we would thus need at least $|M|$ sample points. 
We choose
$s=|M|+8$ deterministic sample points
$x^{(1)},\dots,x^{(s)}\in\Q_{>0}^{\,r}$, built from consecutive primes so
that no two coordinates coincide, to be on the safe side. For each $x^{(i)}$, we can construct a diagonal group element given by
\begin{equation}
    g(x^{(i)})= \begin{pmatrix}
x^{(i)}_1 &     &        &    &  \\
    & (x^{(i)}_1)^{-1} &        & &    \\
    &     & \ddots &     &\\
    &     &        & x^{(i)}_r & \\
    & & & & (x^{(i)}_r)^{-1}
\end{pmatrix},
\end{equation}
which lives in the maximal torus of the complexified group. We are free to evaluate the characters on the complexified group because finite-dimensional representations of the compact group extend to its complexification, and their characters therefore extend as algebraic functions of the complexified torus variables. This lets us evaluate the character functions at rational values rather than the usual complex-phase eigenvalues of the compact groups. For $\SO(2r+1)$, we include an additional eigenvalue $1$ and parametrise $x(y)=y^2$ to avoid half-integer exponents in the character formulas, so we can continue working with rational numbers.

For each torus point $x^{(i)}=(x_1^{(i)},...,x_r^{(i)})$, we evaluate the left-hand side $\GL(m)$ character of \eqref{eq:problem} using the Schur polynomial on the corresponding $m$ eigenvalues
\begin{equation}
\begin{cases}
     \chi_\lambda(g)=s_\lambda(x_1^{(i)},(x_1^{(i)})^{-1},...,x_r^{(i)},(x_r^{(i)})^{-1}), \quad m=2r \\
    \chi_\lambda(g)=s_\lambda(x_1^{(i)},(x_1^{(i)})^{-1},...,x_r^{(i)},(x_r^{(i)})^{-1},1), \quad m=2r+1.
\end{cases}    
\end{equation}
The Schur polynomial encodes the character of the original $\GL(m)$ representation, and is given for arbitrary $\lambda$ by the Weyl character formula
\begin{equation}
    s_\lambda(z_1,...z_m)=\frac{\det\left[\left(z_j^{\lambda_i+m-i}\right)^m_{i,j=1}\right]}{\det \left[\left(z_j^{m-i}\right)_{i,j=1}^m\right]}
\end{equation}
All characters and the Schur polynomial
$s_\lam$ are evaluated as exact ratios of determinants over
$\Q$ using Bareiss fraction-free elimination~\cite{Bareiss1968} which runs in $\mathcal{O}(m^3)$ time. Degenerate
points, where a Weyl denominator vanishes, raise an exception and are
replaced by fresh primes.

The program then solves the overdetermined linear system
$\sum_\mu c_\mu\,\chi_\mu(x^{(i)})=s_\lam(x^{(i)})$, $i=1,\dots,s$ the standard way by solving the normal equations:

\begin{equation}
    A^TAc=A^Tb
\end{equation}
where 
\begin{equation}
    A_{i\mu}=\chi_\mu(x^{(i)}),\quad b_i=s_\lam(x^{(i)}), \quad c_\mu\text{ unknown.}
\end{equation}
Our code then verifies: (i) the solution reproduces
the right-hand side exactly at every sample point and at one further,
independent point; (ii) every $c_\mu$ is an integer (checked through the
exact rational's denominator, not a tolerance); (iii)
$\sum_\mu c_\mu\dim_G(\mu)=\dim\mathbb{S}_\lam(\C^m)$, where $\dim_G$ is the dimension of the irreducible $G$ representation labelled by $\mu$, computed using the Weyl-dimension formula (with the doubling convention of
Lemma~\ref{lem:pair}), given in the Appendix \ref{sec:appendixa}.
Note that even though it seems like our system is overconstrained from the $s=|M|+8$ equations for $|M|$ multiplicities, our additional equations are not independent constraints, and are redundancies for certification. For a correct decomposition, the character identity holds exactly, so we expect a solution with zero residual from the least-squares solver. Failure of any check aborts the computation with an error.

\begin{proposition}[A posteriori correctness]
\label{prop:certified}
If the procedure returns, its output is the true branching multiplicity
vector.
\end{proposition}
\begin{proof}
By Proposition~\ref{prop:parity} the true multiplicities $c^{*}$ satisfy
$s_\lam|_G=\sum_{\mu\in M}c^{*}_\mu\chi_\mu$ identically, hence in
particular at the sample points.  Success of the exact normal-equations
solve means $A^TA$ is nonsingular, i.e.\ the sample matrix
$A_{t\mu}=\chi_\mu(x^{(t)})$ has full column rank; the sampled system
therefore has at most one solution.  The returned vector satisfies it by
check (i), and $c^{*}$ satisfies it identically; hence they coincide.
\end{proof}

\begin{remark}
Checks (i)-(iii) are logically redundant given full column rank, but they
are retained deliberately: they certify the implementation, not only
the mathematics, giving immediate failures rather than data corruption, as expected from a certification algorithm\cite{McConnell2011}.
\end{remark}

The coefficient dimensions in check (iii) calculate the dimension of each $G$- representation using its own $G$-representation formula, not the dimension of the $\GL$ Young diagram, which is generally larger. For example, for the $\lambda=(2,2)$ diagram,
\begin{equation}
    \dim_{\GL(4)}\mathbb{S}_{(2,2)}(\mathbb{C}^4)=20,
\end{equation}
but for the highest weight module $V_{(2,2)}$ of $\Sp(4)$
\begin{equation}
    \dim_{\Sp(4)}V_{(2,2)}=14,
\end{equation}
as $\Sp(4)$ has additional relations coming from the symplectic form.

\section{Explicit Witness Contractions with Invariant Tensors}
\label{sec:witnesses}

The multiplicities certified above answer ``how many'' of each irreducible representation exist in $\mathbb{S}_\lam(\C^m)|_G$.  For the superfield expansion, one also wants, for each copy, an explicit contraction with invariant tensors: a ``witness'' that realises a copy of a certified irreducible representation. The
witness layer supplies one wherever an elementary witness exists, and
reports the remainder honestly as composite, unresolved witnesses.  Crucially, the layer is sound by construction: it emits at most the certified multiplicity of
witnesses per branch, and dimension bookkeeping depends only on the certified
data, so no failure or over-eagerness of the search can affect any count. We group these into three types of witnesses.

\paragraph{Pair witnesses.}
For a branch $\lam\to\mu$ removing $2j=|\lambda|-|\mu|$ boxes by pairings, this corresponds to the removal of $2j$ free indices from the coefficient field. An elementary witness is an explicit contraction of the $2j$ indices into a 
choice of $j$ disjoint index pairs, using the antisymmetric symplectic form $\omega$ ($\Sp$) or symmetric metric
$\eta$ ($\Or/\SO$), chosen to be compatible with the Young diagram symmetrisation. To do this, we consider the skew diagram $\lambda/\mu$ consisting of the removed boxes of the Young tableau. We then perform a Littlewood-Richardson filling of the skew diagram, in which entries are weakly increasing across rows, strictly down columns, with a right-to-left counting ``lattice rule'' which ensures that when reading, the number of $k$s do not exceed the number of $k-1$s.

For example, we could have the skew tableau $\lambda/\mu$ from
\[
\lambda=
\begin{ytableau}
\;&\;&\;&\;\\
\;&\;&\;\\
\;&\;
\end{ytableau}
\qquad
\mu=
\begin{ytableau}
\;&\;\\
\;
\end{ytableau}
\qquad\Longrightarrow\qquad
\lambda/\mu=
\begin{ytableau}
\none & \none & \;&\;\\
\none & \;&\;\\
\;&\;
\end{ytableau},
\]
which removes six boxes, so that $j=3$. We could then consider the filling
\[
\lambda/\mu=\begin{ytableau}
\none & \none & 1&1\\
\none & 2&2\\
3&3
\end{ytableau}
\]
Its rows are weakly increasing, both of its columns increase strictly, and its right-to-left, top-to-bottom reading word is $(1,1,2,2,3,3)$. At every stage, the number of $2$s is less than or equal to the number of $1$s, and the number of $3$s is less than or equal to the number of $2$s, so the lattice condition is satisfied.

This is a bookkeeping rule which provides a candidate contraction pattern that is compatible with the branching $\lambda\to \mu$. The labels in the filling encode how the removed boxes form pairings. This comes from the Littlewood-Richardson theorem, which says that the number of valid fillings of $\lambda/\mu$ with content $\nu$ is the Littlewood-Richardson coefficient $c^\lambda_{\mu\nu}$. Here, the content refers to the tuple $\nu=(\nu_1,\nu_2,\dots)$ where $\nu_k$ is the number of cells carrying label $k$, for instance $\nu=(2,2,2)$ in the above example. Thus, the filling rule dictates the representation-theoretically allowed couplings, which are possible candidates from contractions with invariant tensors.

For the orthogonal group generically the content is $(2,\dots, 2)$. Each pair is assigned a common label so that two boxes with the same label are contracted using the metric $\eta$. Thus, the filling provides an explicit prescription for index contractions. 

For the symplectic case, with $2j$ removed boxes, the content is always of the form $(j,j)$, labeling $j$ boxes with $1$ and $j$ boxes with $2$. Each pair corresponds to one contraction with the symplectic form $\omega$. This would correspond to an alternate filling of the above tableau, with content $(3,3)$
\[\lambda/\mu=\begin{ytableau}
\none & \none & 1&1\\
\none & 1&2\\
2&2
\end{ytableau}\]
Equivalently, labels can be represented with a balanced sequence of opening and closing brackets, where the opening/closing specifies the order of the indices, which is important as the form is antisymmetric. The above filling can also be written as the balanced bracket string ``$(()())$'', obtained from the lattice word $112122$, where ``$1$'' denotes an opening and ``$2$'' a closing bracket. 

Note that the filling of the tableau is dependent on the group under consideration, which determines the content. For explicit comparison, writing $a,\dots,f$ for the six removed cells, we can read off the implied invariant tensors
\[
\begin{ytableau}
\none & \none & a&b\\
\none & c&d\\
e&f
\end{ytableau}
\qquad
\text{symplectic:}\ \ \omega_{ad}\,\omega_{cf}\,\omega_{be}
\qquad
\text{orthogonal:}\ \ \eta_{ab}\,\eta_{cd}\,\eta_{ef}.
\]

\paragraph{$\eps$ witnesses.}
Besides pairing indices with $\omega$ or $\eta$, for $\SL$, $\Sp$ and $\SO$, one may also contract with the totally antisymmetric volume-form, the Levi-Civita tensor $\epsilon$.
This is invariant for the special groups, however only up to sign for $\Or(m)$; accordingly
the following strategies are enabled for $\SO$ and disabled for $\Or$, while $\Sp$ gets a more in-depth treatment.

A column of full height $m$ is contracted with one $\eps$; for
$\Sp(2r)$ this $\eps$ is proportional to $\omega^{\wedge r}$. 
More generally, a partial contraction of $\eps$ may consume a column of height $h<m$, leaving $f=m-h$ mutually antisymmetric free slots on the epsilon, which then attach to the coefficient field, for example
\begin{equation}
    C^{\alpha_1,...\alpha_f}\epsilon_{\alpha_1...\alpha_f\beta_{f+1}...\beta_m}\theta^{\beta_{f+1}}...\theta^{\beta_{m}}.
\end{equation}
Because the remaining indices are completely antisymmetric, the Pieri rule says that they must be added to the Young diagram of the coefficient field in different rows, as they transform under $\Lambda^f V$. This is the explicit contraction underlying the odd-parity case of Proposition~\ref{prop:parity}: for example,
$\Lambda^2\C^3|_{\SO(3)}$ yields
$C^{k}\,\eps_{kij}\,\theta^i\theta^j$, the cross product, which is an independent structure that would not be captured by $\eta$ contractions.

\paragraph{Symplectic $\omega^{\wedge k}$ witnesses.}
For $\Sp(m)$, $m=2r$, the pair and $\epsilon$ witnesses above are the two
extremes of a single family. Because $\omega$ is antisymmetric, its wedge
powers
\begin{equation}
(\omega^{\wedge k})_{i_1\cdots i_{2k}}
\;\propto\;
\omega_{[i_1i_2}\cdots\omega_{i_{2k-1}i_{2k}]},
\qquad 2k\le m,
\end{equation}
are invariant totally antisymmetric tensors of rank $2k$, whose components are
the $2k\times 2k$ Pfaffians of $\omega$. The extreme cases are
$\omega^{\wedge1}=\omega$, the pair witness, and $\omega^{\wedge r}\propto\epsilon$,
the volume form. For $2k<m$ the tensor $\omega^{\wedge k}$ is \emph{not} a
partial contraction of $\epsilon$; it is an independent invariant of its own
rank, and thus must be tracked separately.

A factor $\omega^{\wedge k}$ contracts to something nonvanishing only when its
$2k$ slots are themselves antisymmetrised, so it has to act inside a single
column of the tableau. A strategy therefore assigns to each column $c$ of
$\lambda$, of height $h_c$, an even number $x_c\le\min(h_c,m)$ of cells, taken
from the bottom of that column and contracted with one $\omega^{\wedge x_c/2}$, subject to $\sum_c x_c\le n-|\mu|$ and that the leftover heights $h_c-x_c$ are
weakly decreasing in $c$. Since $x_c$ is even, a column of odd height is never consumed entirely, and the largest admissible strip leaves one cell standing. The weakly-decreasing condition ensures that the
unconsumed cells constitute a genuine Young diagram, whose columns are the top segments of the original columns. The pair search above then applies to what remains, pairing the $2j=n-\sum_c x_c-|\mu|$ residual boxes by a Littlewood-Richardson filling as before.

The strips here are necessary. A pair filling can place at
most one $\omega$ inside any single column. The content $(j,j)$ offers only the
two labels $1$ and $2$, and entries increase strictly down columns, so a column
of $\lambda/\mu$ can hold at most two cells, which is room for the two slots of one
$\omega$, and no more. A factor $\omega^{\wedge k}$ with $k\ge2$ acting inside a
single column would need $2k\ge4$ cells of that column, so no filling can
produce it. The contraction is nonetheless nonzero: those $2k$ slots lie in one
column and are therefore antisymmetrised, which is precisely the condition
under which $\omega^{\wedge k}$ survives. These are the Pfaffian-type invariants and a strip $x_c=2k$ is the only route by which the search can reach them.

The strategies are ordered so that pure pair witnesses are tried first, then
those using fewer and larger strips, and the enumeration is truncated at the
certified multiplicity. As with the other witnesses, the search emits at most $c_\mu$ terms per branch and doesn't influence the dimension count. Strips and pairs together are still not exhaustive. For $\lambda=(3,2,2,1)$
under $\Sp(4)$ the certified copy of $\mu=(1,1)$, of dimension $5$, admits
neither and is still reported as composite.

We note that for the same branching, Gates, Hu and Mak~\cite{Gates:2020gtu} propose graphical ``tying rules'', similar to what we have presented here, as an alternative to Littlewood's Schur-function series, conjectured to be equivalent to Littlewood's, which they explicitly verify for tableaux of at most three columns. The character method here is verifiably correct a posteriori (Proposition~\ref{prop:certified}) for any shape and rank, including for $\Sp$.

\section{Complexity}
\label{sec:complexity}
Here we consider the runtime complexity of our algorithm.
\begin{proposition}
\label{prop:complexity}
Consider a candidate list of cardinality $|M|$.
For one Cauchy block $\mathbb{S}_\lam(V)\otimes\mathbb{S}_{\lam'}(W)$ of degree $|\lambda|=n$ at rank $r$, the branching computation
performs $O\!\big(|M|^2\big)$ character and $O\!\big(|M|\big)$ Schur evaluations. Each is a
determinant of size at most $m$, costing $O(m^{3})$ ring operations on
rationals of bit length $O\!\big((n+m)\log P\big)$ (Bareiss intermediates grow this by a further factor of $m$), where
$P=O\!\big(|M|\,r\log(|M|\,r)\big)$ bounds the primes used.
The exact linear algebra adds $O(|M|^3)$ rational operations. Finally
\begin{equation}
\label{eq:Mbound}
|M|\;=\;
\begin{cases}
\displaystyle\sum_{\substack{j\le n\\ n-j\ \mathrm{even}}} p_{\le r}(j),
  & m=2r,\\[14pt]
\displaystyle\sum_{j\le n} p_{\le r}(j),
  & m=2r+1,
\end{cases}
\qquad\text{and in either case}\quad |M|\le(n{+}1)\,p(n),
\end{equation}
with $p(k)$ the number of partitions
of $k$. The total cost is thus polynomial in $p(n)$ and $r$, and depends on
$m$ only through $r$ and the parity of $m$.
\end{proposition}
The bound is polynomial in $p(n)$ rather than in $n$. Since $p(n)$ grows
subexponentially this is not a polynomial-time algorithm in the degree. This is not as bad as it seems practically, since $n\le d_Sd_F$ and the partition numbers stay relatively small in that range, eg $p(8)=22$, $p(16)=231$. The parity dependence is crucial here. At even $m$ the candidate list
runs over $n,n-2,\dots$, whereas at odd $m$ Proposition~\ref{prop:parity}
admits both parities and $|M|$ roughly doubles. Thus, the odd parity cases are measurably slower.

The salient point is what the bound does not contain: no
Littlewood-Richardson coefficient is ever computed to obtain a
multiplicity, no modification rule is applied, and no tensor space $(\mathbb{C}^m)^{\otimes n}$ of
dimension $m^n$ is materialised. Fillings appear only in the witness layer of Section~\ref{sec:witnesses}, where they propose explicit contractions. 

It should be said plainly that this is not a speed advantage in most cases of interest. In the stable
range $\ell(\lam)\le r$, where Littlewood's rule applies directly, evaluating it
is far cheaper than the character solve. On the small shapes considered here, the
classical route returns the same result much faster than character comparison. The cost comes from  the fact that Littlewood's rule evaluates the multiplicities directly, whereas we must invert an overdetermined system which is exacerbated by the cost of rational arithmetic, since the sample entries are powers $p^{e}$ with $e$ growing as $n+m$, so their bit lengths grow as $(n+m)\log P$. We note that computing the Littlewood-Richardson coefficient is a \#P-complete\cite{narayanan2005computationkostkanumberslittlewoodrichardson} problem regardless, though it is not an issue at the small degrees considered here.

What the method buys instead is uniformity and verifiability. A single code path
covers $\Sp(2r)$, $\SO(2r)$ and $\SO(2r+1)$, both inside and outside the stable
range, with no case analysis, no signed cancellations, and no modification-rule
calculus to implement correctly. The answer is also self-certifying, as an incomplete candidate set or implementation mistake would yield an inconsistent linear system, rather than a plausible but incorrect answer. For a tool whose output is intended to be used without further checking we regard that as a worthwhile trade, but it is a trade. 


The witness layer performs at worst $p(j)$ backtracking searches whose worst case is exponential in $n$. Since $n\le d_S d_F$ is small in applications and the search is truncated at the certified multiplicity, it is negligible in practice, and, by the
soundness property of Section~\ref{sec:witnesses}, it cannot affect
correctness under any circumstances.

Note that the prime number sample points are deterministic, and are fixed by $|M|$ and $r$ rather than sampled randomly.  This means that every run of every check
is exactly reproducible, and a reported failure is a permanent, shareable
artefact.

To verify that runtime of the decomposition does not scale with the dimension of the representation, we consider the branching to $\SO(10)$ corresponding to rank $r=5$. For fixed degree $n=4$, presented in Table \ref{tab:dimscaling}, we see timings barely vary, and in fact are faster for the largest dimensional case than the smallest dimensional one when tested on our machine. 
\begin{table}[ht]
\centering
\begin{tabular}{lrr}
\toprule
$\lam$ & $\dim\mathbb{S}_\lam(\C^{10})$ & time (s)\\
\midrule
$(1,1,1,1)$ &  210 & 0.064\\
$(4)$       &  715 & 0.055\\
$(2,2)$     &  825 & 0.053\\
$(2,1,1)$   &  990 & 0.068\\
$(3,1)$     & 1485 & 0.056\\
\bottomrule
\end{tabular}
\caption{Branching to $\SO(10)$ at fixed degree $n=4$ with  $|M|=8$ fixed. Timings are medians over 20 runs.}
\label{tab:dimscaling}
\end{table}
This provides a contrast to branching by projecting with weights which requires enumeration or explicit construction of projectors. The character method here never forms a weight or a basis
vector.

\section{Validation and Examples}
\label{sec:validation}

Three complementary layers of validation were applied.

\emph{(a) Internal certification.} We check the calculation from inside the program. Every branching in every run is
certified as in Section~\ref{sec:branching}, checking integer-valueness, dimensionality, etc. The top-level renderer
additionally checks, at each Grassmann order $n$, the sum rule
$\sum_{\text{structures}}\dim C_S\cdot\dim C_F=\binom{d_S d_F}{n}$, which counts the expected number of components. This is a necessary condition for a consistent decomposition.

\emph{(b) Cross-implementation sweep.}  The branching algorithm was ported
independently to exact-fraction arithmetic (meaning we should see an exact match, not modulo machine precision) in a second language (Python) and swept
over all shapes with $n\le6$ and all $m\in\{2,\dots,8\}$ (185 shape/rank
pairs), with every internal check passing, and with the even-parity-only
ansatz reproducibly failing at odd $m$ exactly as
Proposition~\ref{prop:parity} predicts.

\emph{(c) Known answers.}  Table~\ref{tab:known} lists hand-checkable cases,
all reproduced.

\begin{table}[ht]
\centering
\begin{tabular}{llll}
\toprule
$\lam$ & $G$ & result & comment\\
\midrule
$(1,1)$ & $\SO(3)$ & $(1)$ & $\Lambda^2\C^3\cong\C^3$ via $\eps_{ijk}$\\
$(1,1,1)$ & $\SO(3)$ & $\varnothing$ & $\Lambda^3\C^3\cong\C$ via $\eps_{ijk}$\\
$(2,1)$ & $\SO(3)$ & $(2)\oplus(1)$ & $8=5\oplus3$, $\mathrm{Sym}_0^2(\mathbb{C}^3)\oplus \mathbb{C}^3$\\
$(2,2)$ & $\SO(3)$ & $(2)\oplus\varnothing$ & $6=5\oplus1$, $\mathrm{Sym}^2_0(\mathbb C^3)\oplus\mathbb C$\\
$(1,1,1)$ & $\SO(4)$ & $(1)$ & $\Lambda^3\mathbb{C}^4\cong \mathbb{C}^4$ via $\epsilon_{ijkl}$\\
$(1,1,1)$ & $\SO(5)$ & $(1,1)$ & $\Lambda^3\C^5\cong\Lambda^2\C^5$, using $\epsilon_{ijklm}$\\
$(2)$ & $\SO(4)$ & $(2)\oplus\varnothing$ & $10=9\oplus1$, $\mathrm{Sym}^2_0(\mathbb C^4)\oplus\mathbb C$\\
$(1^4)$ & $\Sp(4)$ & $\varnothing$ & Volume form $\eps\propto\omega\wedge\omega$\\
\bottomrule
\end{tabular}
\caption{Hand-verifiable branchings reproduced by the algorithm.  $\varnothing$
denotes the trivial representation.}
\label{tab:known}
\end{table}

We consider now some simple examples of interest: 
\begin{example}[$\SL(2)\times\SO(3)$]
For $d_S=2$, $d_F=3$ the order-$2$ component coefficient field is contracted with a term of the form
\begin{equation}
    C_{i_1i_2}^{\alpha_1\alpha_2}\theta^{i_1}_{\alpha_1}\theta^{i_2}_{\alpha_2},
\end{equation}
where $\alpha_k=1,2$ are $\SL(2)$ spinor indices, and $i_k=1,2,3$ are $\SO(3)$ flavour indices. The total number of components at this order is $\binom{6}{2}=15$. This decomposes as
\begin{equation}
    C^{(\alpha_1\alpha_2)k}
\epsilon_{k i_1i_2}
\theta^{i_1}_{\alpha_1}\theta^{i_2}_{\alpha_2}
\qquad
3\times3=9,
\end{equation}
\begin{equation}
C_{(i_1i_2)}  \epsilon^{\alpha_1\alpha_2}
\theta^{i_1}_{\alpha_1}\theta^{i_2}_{\alpha_2},
\qquad
1\times5=5 \text{ (traceless)},
\end{equation}
and
\begin{equation}
    C\epsilon^{\alpha_1\alpha_2}
\eta_{i_1i_2}
\theta^{i_1}_{\alpha_1}\theta^{i_2}_{\alpha_2},
\qquad
1\times 1=1.
\end{equation}
The dimensions
sum to $15$ as certified. The first term is the $\eps$ contraction of the odd-parity case of
Proposition~\ref{prop:parity}, associated to the $\SO(3)$ branching $\Lambda^2\mathbb C^3\cong\mathbb C^3$ explicitly.
\end{example}

\begin{example}[$6D$ $\mathcal{N}=2$ Koller Decomposition]
Koller in \cite{Koller:1982cs} catalogued by hand the complete algebra of spinorial
derivatives on $D=6$, $\mathcal{N}=2$ superspace. He uses these as projectors to determine the independent structures. To convert from the notation of \cite{Koller:1982cs} to that of this paper, the
spinor index in $\mathrm{SU}^*(4)$ with indices $A,B,C,...$ corresponds to the real form of our $G_S=\SL$ with $d_S=4$ and indices $\alpha_k$, and the isospin index of $\mathrm{USp}(2)$ with indices $a,b,c,...$ is captured by $G_F=\Sp$ and $d_F=2$ with indices $i_k$ in our notation.

His Section~3.1 derives the first layer of our algorithm by hand by applying the
Young projector of a tableau to the $\mathrm{SU}^*(4)$ indices, which is the Cauchy
decomposition~\eqref{eq:cauchy}. His requirement that the two tableaux
carry at most $d_S=4$ and $d_F=2$ boxes per column is the box condition of
Section~\ref{sec:cauchy}. What he then performs case by case (raising and
lowering with $\epsilon$ for both groups, and removing $\mathrm{USp}$ traces) is our second layer. Running the algorithm at each order reproduces the catalogue of his eq. (5) in full, seen in Table \ref{tab:koller}. 

\begin{table}[ht]
\centering
$
\begin{array}{c|l|c|l}
n & \SL(4)\times\Sp(2) \text{ labels with [dimension]} & \textstyle\binom{8}{n} & \text{Koller~(5)}\\
\hline
0 & \varnothing\times\varnothing\ [1] & 1 & 1\\
1 & (1)\times(1)\ [8] & 8 & D^{(1)}\\
2 & (2)\times\varnothing\ [10],\ (1,1)\times(2)\ [18] & 28 & D^{(2)}\\
3 & (2,1)\times(1)\ [40],\ (1,1,1)\times(3)\ [16] & 56 & D^{(3)}\\
4 & (2,2)\times\varnothing\ [20],\ (2,1,1)\times(2)\ [45],\ \varnothing\times(4)\ [5] & 70 & D^{(4)}\\
5 & (2,2,1)\times(1)\ [40],\ (1)\times(3)\ [16] & 56 & D^{(5)}\\
6 & (2,2,2)\times\varnothing\ [10],\ (1,1)\times(2)\ [18] & 28 & D^{(6)}\\
7 & (1,1,1)\times(1)\ [8] & 8 & D^{(7)}\\
8 & \varnothing\times\varnothing\ [1] & 1 & D^{(8)}
\end{array}$
\caption{$6D$ $\mathcal{N}=2$ decomposition.}
\label{tab:koller}
\end{table}

Every order closes against $\binom{8}{n}$, and the orders total $256=2^{d_Sd_F}$.
To illustrate further, at $n=4$ the algorithm returns three structures, in one-to-one correspondence
with Koller's level-four operators:
\begin{equation}
\begin{array}{lcll}
\Cyd{\yd{2,2}}{}^{\alpha_1\alpha_2\alpha_3\alpha_4}\,\epsilon_{i_1i_2}\epsilon_{i_3i_4}
  &\longleftrightarrow& D^{(4)}_{AB}{}^{CD}, & 20\times1=20,\\[4pt]
\Cyd{\yd{2,1,1}}{\yd{2}}{}^{\alpha_1\alpha_2\alpha_3\alpha_4}_{i_3i_4}\,\epsilon_{i_1i_2}
  &\longleftrightarrow& D^{(4)}_{ab\,A}{}^{B}, & 15\times3=45,\\[4pt]
\Cyd{}{\yd{4}}{}_{i_1i_2i_3i_4}\,\epsilon^{\alpha_1\alpha_2\alpha_3\alpha_4}
  &\longleftrightarrow& D^{(4)}_{abcd}, & 1\times5=5,
\end{array}
\end{equation}
so the level closes at $20+45+5=70=\binom{8}{4}$. We annotate here our coefficients with the Young tableaux of the remaining indices for clarity, filled left to right in index order. Each of Koller's operators pick out exactly one of these structures. For example, acting on a superfield and evaluating at $\theta=0$, $D^{(4)}_{abcd}$ returns the coefficient $C_{(i_1i_2i_3i_4)}$, the contraction with $\epsilon^{\alpha_1\alpha_2\alpha_3\alpha_4}$ implicit in
his definition of the operator.
\end{example}

\begin{example}[3D $\mathcal{N}=4$: $\SL(2)\times\SO(4)$]
\label{ex:3dN4}
Three-dimensional $\mathcal{N}$-extended superspace has $\theta^I_\alpha$ with
$\alpha=1,2$ a Majorana index of $\SL(2,\mathbb{R})$ and $I=1,\dots,\mathcal{N}$
the vector index of $\SO(\mathcal{N})_R$, so that $d_S=2$ with $G_S=\SL$ and
$d_F=\mathcal{N}$ with $G_F=\SO$. At $\mathcal{N}=4$ the full expansion carries
$2^{8}=256$ components. Contracting each term against
$\theta^{i_1}_{\alpha_1}\cdots\theta^{i_4}_{\alpha_4}$, the order-$4$
coefficient space of dimension $\binom{8}{4}=70$ decomposes as
\begin{equation}
\begin{array}{lcll}
\Cyd{\yd{4}}{}{}^{\alpha_1\alpha_2\alpha_3\alpha_4}\,\epsilon_{i_1i_2i_3i_4}
   &\qquad 5\times1=5,\\
\Cyd{\yd{2}}{\yd{1,1}}{}^{\alpha_2\alpha_3}_{i_1i_2}\,
   \epsilon^{\alpha_1\alpha_4}\eta_{i_3i_4}
   &\qquad 3\times6=18,\\
\Cyd{\yd{2}}{\yd{2}}{}^{\alpha_2\alpha_3}_{i_4m_1}\,
   \epsilon^{\alpha_1\alpha_4}\epsilon_{i_1i_2i_3m_1}
   &\qquad 3\times9=27,\\
C\epsilon^{\alpha_1\alpha_3}\epsilon^{\alpha_2\alpha_4}\eta_{i_1i_3}\eta_{i_2i_4}
   &\qquad 1\times1=1,\\
\Cyd{}{\yd{2,2}}{}_{i_1i_3i_2i_4}\,
   \epsilon^{\alpha_1\alpha_3}\epsilon^{\alpha_2\alpha_4}
   &\qquad 1\times10=10,\\
\Cyd{}{\yd{2}}{}_{i_1i_3}\,
   \epsilon^{\alpha_1\alpha_3}\epsilon^{\alpha_2\alpha_4}\eta_{i_2i_4}
   &\qquad 1\times9=9,
\end{array}
\end{equation}
summing to $70$ as certified. Two features appear here that the above examples do not exhibit. The third term uses a partial $\epsilon$ contraction, its unused slot $m_1$ attaching to the coefficient as in Section~\ref{sec:witnesses}. Two coefficients carry $\SO(4)$ labels of full length: $[i_1i_2]$ of dimension $6$ is contained in the second piece, and
$[(i_1i_3)(i_2i_4)]$ of dimension $10$ in the fifth piece. These each are the pairs $\mu_+\oplus\mu_-$ of
Lemma~\ref{lem:pair}, being $3+3$ and $5+5$ respectively corresponding to the parts transforming under each copy of $\SU(2)$ in $\SO(4)\cong (\SU(2)\times \SU(2))/\mathbb{Z}_2$.
\end{example}

\section{Limitations and Outlook}
\label{sec:limitations}

Three limitations are inherent to the current design of the $\GL$ branching and witness layer and are exposed,
rather than hidden, by the implementation.  (1) Characters are sampled on
the identity component, so the ``$\Or$'' option computes $\SO$ content. Two distinct $\Or$ representations can become identical when restricted to $\SO$: two representations may agree on $\SO(m)$ but
differ by the determinant character on the disconnected component. Genuine
$\Or(m)$ decomposition would require sampling reflections, which is a
straightforward extension of the same framework.  (2) At even dimension, the two
irreducibles of a length-$r$ label are treated as one object of doubled
dimension $\mu_+\oplus\mu_-$ (Lemma~\ref{lem:pair}); separating the self-dual $\mu_+$ and anti-self-dual $\mu_-$
coefficients would be a possible refinement.  (3) Not every branch has an explicit elementary witness. If an elementary witness cannot be found using $\omega^{\wedge k}$, $\eta$, and $\epsilon$ for a branch copy, the copy is reported as composite with certified shape and
dimension but no explicit contraction. Extending the witness grammar could produce expressions for some of these. However, the multiplicities will be the same regardless, as the witness layer is only a supplement to the branching calculation.

More broadly, the architecture of an exact, deterministic computation,
wrapped in run-time proofs of its own output, can convert both mathematical gaps and
implementation faults into immediate, diagnosable failures.  We suggest it as a template for representation-theoretic software in physics generally. 

We note that further optimisations to the speed of this code are possible using modular arithmetic to alleviate some of the burden of the rational arithmetic. Early tests suggest many orders-of-magnitude improvements for large degrees. This will be explored and certified in a future version.

Our procedure is limited in the sense that it only applies to spinors in six dimensions or lower. This is because generically, the spin groups do not admit isomorphisms to the classical Lie groups. An important direction for future development would be to extend this approach to the full spin groups, allowing for applications to the study of superspace in 10 and 11D.

We note also that in many cases Grassmann coordinates come in complex conjugate pairs $(\theta,\bar\theta)$. We here consider only the decomposition of monomials in $\theta$, however, in monomials consisting of both $\theta$ and $\bar\theta$, there may be additional reductions depending on the group. This is enhanced by the fact that spinors can have non-trivial properties under complex conjugation. Furthermore, in some parametrisations of superspace, the Grassmann numbers can carry even more indices. We do not account for this possibility here, and rather assume the indices transform under the fundamental representation. These provide some possible avenues for extension of the current script.

\paragraph{Code availability.}
The implementation is a self-contained Julia module (approximately a thousand lines, no dependencies beyond the standard library) available on Gitlab \href{https://gitlab.com/B0bGary/grassmann-monomials}{here}. It can also be found included in the Julia package ``\href{https://gitlab.com/B0bGary/alakazam.jl}{Alakazam}''\cite{woods2026alakazam}, where it is supplemented by explicit functions to generate superfield expansions.

\section*{Acknowledgements}
The author would like to thank Gabriele Tartaglino-Mazzucchelli for useful comments, especially pertaining to existing literature.
The author is supported by the Swiss National Science Foundation under grant number 200021\_219267. The author would like to thank the Ramsay Centre for Western Civilisation for its support.

\appendix\section{Weyl Dimension Formulas}\label{sec:appendixa}
For an irreducible representation with highest-weight label
\(\mu=(\mu_1,\ldots,\mu_r)\), the Weyl dimension formula gives its
dimension directly from the highest weight.  In the conventions used here,
the relevant formulas are:

For \(G=\Sp(2r)\),
\begin{equation}
\label{eq:weyl-dim-sp}
\dim_{\Sp(2r)}(\mu)
=
\prod_{i=1}^r
\frac{\mu_i+r-i+1}{r-i+1}
\prod_{1\leq i<j\leq r}
\frac{
(\mu_i-\mu_j+j-i)
(\mu_i+\mu_j+2r-i-j+2)
}{
(j-i)(2r-i-j+2)
}.
\end{equation}

For \(G=\SO(2r+1)\),
\begin{equation}
\label{eq:weyl-dim-so-odd}
\dim_{\SO(2r+1)}(\mu)
=
\prod_{i=1}^r
\frac{\mu_i+r-i+\frac12}{r-i+\frac12}
\prod_{1\leq i<j\leq r}
\frac{
(\mu_i-\mu_j+j-i)
(\mu_i+\mu_j+2r-i-j+1)
}{
(j-i)(2r-i-j+1)
}.
\end{equation}

For \(G=\SO(2r)\),
\begin{equation}
\label{eq:weyl-dim-so-even}
\dim_{\SO(2r)}(\mu)
=
\prod_{1\leq i<j\leq r}
\frac{
(\mu_i-\mu_j+j-i)
(\mu_i+\mu_j+2r-i-j)
}{
(j-i)(2r-i-j)
}.
\end{equation}

Equation~\eqref{eq:weyl-dim-so-even} is the dimension of a single irreducible
representation. By the convention fixed in Lemma~\ref{lem:pair}, a label with
$\ell(\mu)=r$ and $\mu_r>0$ denotes throughout the \emph{pair} $\mu_\pm$, so
the dimension entering the sum rule of check~(iii) is
\begin{equation}
\label{eq:weyl-dim-so-even-pair}
\dim_{\SO(2r)}\!\big(\mu_+\oplus\mu_-\big)
=2\,\dim_{\SO(2r)}(\mu),
\qquad \ell(\mu)=r,\ \mu_r>0,
\end{equation}
and $\dim_{\SO(2r)}(\mu)$ as given by~\eqref{eq:weyl-dim-so-even} in every
other case. For example, $\mu=(1,1)$ at $r=2$
gives $\dim_{\SO(4)}(\mu)=3$, hence $6$ for the pair, which is
$\dim\Lambda^2\C^4$ as required.

\newpage
\printbibliography
\end{document}